\documentclass[copyright,creativecommons]{eptcs}
\providecommand{\event}{LFMTP 2011}

\usepackage{amsmath}
\usepackage{amssymb}
\usepackage{amsthm}
\usepackage{amsfonts}
\usepackage{proof}
\usepackage{stmaryrd}
\usepackage{beluga}
\usepackage{colonequals}
\usepackage{proof}
\usepackage{enumerate}

\newcommand{\var}[1]{\# 1}

\ifdefined\LONGVERSION
  \relax
\else
\newcommand{\LONGVERSION}[1]{}

\fi

\usepackage{srcltx}

\usepackage{listings}
\newdimen\zzlistingsize
\newdimen\zzlistingsizedefault
\zzlistingsize=\zzlistingsizedefault
\global\def\InsideComment{0}
\newcommand{\Lstbasicstyle}{\fontsize{\zzlistingsize}{1.05\zzlistingsize}\ttfamily}
\newcommand{\keywordFmt}{\fontsize{0.9\zzlistingsize}{1.0\zzlistingsize}\bf}
\newcommand{\smartkeywordFmt}{\if0\InsideComment\keywordFmt\fi}
\newcommand{\commentFmt}{\def\InsideComment{1}\fontsize{0.95\zzlistingsize}{1.0\zzlistingsize}\rmfamily\slshape}

\newlength{\zzlstwidth}
\newcommand{\setlistingsize}[1]{\zzlistingsize=#1%
\settowidth{\zzlstwidth}{{\Lstbasicstyle~}}}
\setlistingsize{\zzlistingsizedefault}

\newcommand{\der}{\,\vdash}

\def\lv{\mathopen{{[\kern-0.14em[}}}    
\def\rv{\mathclose{{]\kern-0.14em]}}}   

\newcommand{\dens}[1]{\mathopen{[\kern-0.3ex[}#1\mathclose{]\kern-0.3ex]}}
\newcommand{\denk}[2]{\mathopen{\{\kern-0.3ex|}#1\mathclose{|\kern-0.3ex\}}_{#2}}
\def\lo{\mathopen{{\lceil\kern-0.25em\lceil}}}    
\def\ro{\mathclose{{\rfloor\kern-0.25em\rfloor}}}

\def\ltox#1{\buildrel\raise1pt\hbox{$\scriptstyle#1$}\over\longrightarrow}

\def\tocolow{\buildrel\raise-5pt\hbox{$\scriptscriptstyle+$}\over\rightarrow}

\newcommand{\mgoal}[1][]{\mbox{goal\ifthenelse{\equal{#1}{}}{}{~#1}}}

\newcommand{\fun}[1]{\Pi#1.\,}

\newcommand{\lam}[1]{\lambda#1.\,}
 \newcommand{\app}[2]{#1 \; #2}

\newcommand{\EL}{\mathcal{E}\kern-0.2ex\ell}

\newcommand{\x}{\mathsf{x}}
\newcommand{\xdel}[1][]{\ifthenelse{\equal{#1}{}}{\x_\Delta}{\x_{\Delta+#1}}}

\newcommand{\VDash}{\mathrel{\mathord{|}\kern-0.15ex\mathord{\models}}}

\newcommand{\bnfas}{\coloncolonequals}
\newcommand{\bnfalt}{\mid}

\newenvironment{rules}{\begin{displaymath}\begin{array}{c}}{\end{array}\end{displaymath}}

\ifdefined\LONGVERSION
  \relax
\else
\newcommand{\LONGVERSION}[1]{}

\fi

\title{Multi-level Contextual Type Theory}
\author{Mathieu Boespflug \quad\quad Brigitte Pientka
\institute{Mcgill University\\ Montreal, Canada}
\email{\{mboes,bpientka\}@cs.mcgill.ca}
}
\def\titlerunning{Multi-level Contextual Type Theory}
\def\authorrunning{M. Boespflug \& B. Pientka}
\begin{document}
\maketitle

\begin{abstract}
  Contextual type theory distinguishes between bound variables and meta-variables to
write potentially incomplete terms in the presence of binders. It has
found good use as a framework for concise explanations of higher-order
unification, characterize holes in proofs, and in developing a
foundation for programming with higher-order abstract syntax, as
embodied by the programming and reasoning environment Beluga. However,
to reason about these applications, we need to introduce
meta$^2$-variables to characterize the dependency on meta-variables
and bound variables. In other words, we must go beyond a two-level
system granting only bound variables and meta-variables.

In this paper we generalize contextual type theory to $n$ levels for
arbitrary $n$, so as to obtain a formal system offering bound
variables, meta-variables and so on all the way to meta$^n$-variables.
We obtain a uniform account by collapsing all these different kinds of
variables into a single notion of variabe indexed by some level $k$.
We give a decidable bi-directional type system which characterizes
$\beta\eta$-normal forms together with a generalized substitution
operation.

\end{abstract}

\section{Introduction}

A core problem when describing computations and proofs is the need to model
unknown entities. The standard approach is to introduce meta-variables that
one can use in place of concrete evidence that might not yet be
available. Consider for
example the development of the proof for $\forall x.\exists y. P(x, x) \wedge Q
(x, x) \supset Q(y,x) \wedge P(x,y)$ in a proof assistant. Working from the goal
formula, we first introduce a parameter $a$ and subsequently introduce a
meta-variable $Y$ for $y$ which may depend on $a$. We can describe the
intermediate subgoal we must now solve as: $P(a,a) \wedge Q(a,a) \supset Q(Y a,
a) \wedge P(Y a, a)$. At a later point in the proof, we may realize
through (higher-order) unification that $a$ is a good instantiation
for $Y$ giving us an trivially provable goal. The question we address in
this paper is how to describe formally the incomplete proof state we are in prior to
finding instantiations for $Y$. Clearly, the missing proof term we want to
construct depends on the meta-variable $Y$ and the parameter $a$ bound in
the context.
We hence need to introduce meta$^2$-variables to describe it.

A similar situation arises in the Beluga programming and reasoning
environment \cite{Pientka:POPL08,Pientka:IJCAR10}. Recursive programs in Beluga
analyze and manipulate meta-objects of type $A[\Psi]$, i.e. objects which
have type $A$ in a bound variable context $\Psi$. For example, 
\lstinline![x:i] allI \y. andI (F x y) (F y x)! describes the derivation of the
formula $\forall y.P(y,x) \wedge P(x,y)$ where \lstinline!F! itself
stands in lieu of a description of the
derivation which ends in $P(y,x)$ in the context
\lstinline![x:i,y:i]!. Note that the meta-variable \lstinline!F! is bound: if
we pattern match on the LF object, then \lstinline!F! is introduced and bound in
the branch or it is bound explicitely at the outside by an abstraction. We hence
have two different kinds of bound variables in the LF object. In Beluga, we can
write underscores anywhere in an LF object and let  type-reconstruction find the
correct instantiation. For example, to describe an incomplete derivation where
we omit the second argument to \lstinline!andI!, we may write 
\lstinline![x:i] allI \y. andI (F x y) _ !. During type reconstruction, the
underscore will be replaced by a meta$^2$-variable to express the fact that we
may use the meta-variable \lstinline!F! or the bound variables \lstinline!x! and
\lstinline!y!. 

Contextual type theory \cite{Nanevski:ICML05} provides both bound
variables and meta-variables, complete with a logical foundation for reasoning
about them. Up to now it has been used to explain higher-order
unification \cite{PientkaPfenning:CADE03,Abel:TLCA11}, characterize
concisely holes in proofs, and develop a foundation for programming with
higher-order abstract syntax as found in the Beluga programming and reasoning
environment \cite{Pientka:POPL08,Pientka:IJCAR10}.
This paper generalizes and extends contextual type theory
\cite{Nanevski:ICML05} to an arbitrary number of levels of
variables. Bound variables are of level $0$, meta-variables are of level $1$,
meta$^2$-variables are of level $2$, and so on and so forth. This
leaves us with a uniform treatment
of contexts, variables and their associated substitution operations.
Unlike earlier work sketched by Pfenning \cite{Pfenning:TLCA07} for
the simply typed case, we enforce that the
context is ordered, i.e.\ if $n > m$, then variables of level $n$ occur to the 
left of variables of level $m$. This will naturally enforce the correct
dependency: variables of the higher level $n$ cannot depend on the variables of
lower level $m$. We give a bi-directional type system to characterize
$\beta$-$\eta$-long normal forms and generalize the hereditary substitution
operation to variables of arbitrary level. We prove the hereditary substitution to be
terminating, prove that typing preserves the well-formedness of ordered
contexts, and show bi-directional typing to be decidable for the multi-level
system. 

This work is one step of the way towards streamlining and simplifying
the implementation of Beluga, where we currently distinguish between
bound variables, meta-variables, and meta$^2$-variables. But more
generally, this work can be used to formalize incomplete proofs that
manipulate open proof objects containing meta-variables. This is
important to scale tactic languages 
such as VeriML \cite{Stampoulis:ICFP10} where we manipulate meta-objects
that may contain bound variables, or to reason about the
tactics themselves. We envision down the line a multi-level Beluga,
which would allow us to reason about and manipulate Beluga programs
within Beluga itself. This will provide a uniform framework where the
proofs, the development of proofs using tactics, and the reasoning
about tactics all share a common basis and supporting implementation.


\section{Language definition}
\subsection{Syntax}
Contextual type theory was introduced by Nanevski et al
\cite{Nanevski:ICML05} and extended the logical framework LF
\cite{Harper93jacm} with first-class meta-variables. Our work is a
natural continuation of this work generalizing contextual types to
multiple levels. Following Watkins et al \cite{Watkins02tr}, the
syntax
is limited to expressing terms in $\beta$-normal forms, which are
sufficient for encoding the types and expressions of some logic
or programming language as well as judgements and derivations
pertaining to those types and expressions. We leave the development of
a non-canonical version to future work. While the grammar below only
enforces that objects are $\beta$-normal, the typing rules will also
ensure objects are moreover in $\eta$-long form.

\[
\begin{array}{lrcl}
  \mbox{Sorts} & s & \bnfas & \lftype \bnfalt \lfkind
  \\
  \mbox{Atomic Types/Kinds} & P, Q & \bnfas & s \bnfalt \const a
  \bnfalt \app{P}{(\hat\Gamma.N)}
  \\
  \mbox{Normal Types/Kinds} & A,B,K & \bnfas & P \bnfalt \fun{x^n{:}A[\Psi^n]}B
  \\
  \mbox{Atomic Terms} & R & \bnfas & x^n[\sigma] \bnfalt \const c
  \bnfalt \app R{(\hat\Gamma.N)}
  \\
  \mbox{Normal Terms} & M,N & \bnfas & R \bnfalt \lam{x^n}{M}
  \\
  \mbox{Substitutions} & \sigma, \tau & \bnfas & \edot \bnfalt \sigma, \hat\Gamma^n.M \bnfalt \sigma, \ctxvarex{x^n}
  \\
  \mbox{Contexts} & \Psi, \Phi, \Gamma & \bnfas & \edot \bnfalt \Psi, x^n{:}A[\Phi^n]
  \\
  \mbox{Signature} & \Sigma & \bnfas & \edot \bnfalt \Sigma, \const a{:}K \bnfalt \Sigma, \const c{:}A
\end{array}
\]

Normal objects may contain variables $x^n$ which are bound by
$\lambda$-abstraction or declared in a context $\Psi$. Variables are
associated with a level $n$. The level $n$ of a context $\Psi$ is an
upper bound on $\{ k+1 | x^k \in \dom(\Psi) \}$. In other words, we
write $\Psi^n$ when we know that all variables in $\Psi$ are at levels
strictly smaller than $n$. Unlike the superscript on variables, the
superscript on contexts is purely a mnemonic convenience as this
annotation can normally be inferred from information elsewhere wherever
relevant, just as lambda-abstractions are not annotated with the
domain type because that information is usually already available
somewhere else.

A variable $x^n$ has type $A[\Psi^n]$, i.e. it has type $A$ in the context
$\Psi$ of variables at levels lower than $n$. To put it differently, $x^n$ may refer to the (local) variables in
$\Psi^n$ and may also contain (global) variables of a higher level. If $n=0$, we
recover our ordinary bound variables of type $A$. The context will be dropped,
because there is no context of level $0$.  Yet, $A$ may refer to meta-variables
or more generally, to variables at a higher level. Similarly, we can recover
meta-variables which are of level $1$ and have type $A[\Psi^1]$. The context
$\Psi^1$ contains only variables of level $0$, i.e. ordinary bound
variables. Hence, locally meta-variables depend on ordinary bound variables, but
they may also contain (global) variables of higher level (for example,
meta$^2$-variables).
A variable $x^n$ of type $A[\Psi^n]$ stands for an object $\Psihat^n. M$ where
$\Psihat$ lists the bound variables that may occur in $M$ and again, the level $n$
indicates that it may contain locally bound variables only up to level $n$. This
is important information when the need arises to rename the locally
bound variables occurring in $M$, to avoid captures for instance.

As we navigate under binders, it may be necessary to substitute the
bound variable for another one or for a term. Such substitutions get
``stuck'' at the level of meta-variables because there is no term to
substitute in until the meta-variable is instantiated. In the core
syntax we impose the
invariant that all meta-variables standing for a term in a context
$\Psi$ be associated with a (simultaneous) substitution $\sigma$
such that the domain of $\sigma$
matches that of $\Psi$. As such, given $\Psi$, it is not necessary to
make the domain of substitutions explicit, as that information would
be redundant. Intuitively, the $i$-th element in $\sigma$ corresponds to
the $i$-th assumption in $\Psi$. A postponed substitution $\sigma$ is
applied as soon as we know what $x^n$ stands for and applying
$\sigma$ with domain $\Psi$ to a term $M$ is written
$[\sigma]_{\Psi}M$.

A substitution maps (canonical) terms for variables. But as we push
this substitution under binders, the size of the context grows and so
must that of the substitution. A term of the form
$(\lam{x}{M})[\sigma]$ can be rewritten to $(\lam{x}{M[\sigma']})$,
where $\sigma'$ extends $\sigma$ by mapping the variable $x$ to
itself. However, recall that $x$ by itself is not a meaningful term in
our grammar --- all variables are systematically paired with a
simultaneous substitution.
Without knowing the type for $x$, we cannot
infer the appropriate identity substitution which would allow us to
replace $x$ by $x[\id]$. Even if $x[\id]$ was made a valid syntactic
object in our grammar, it is not guaranteed that $x[\id]$ is a
canonical form at higher type (canonical forms are $\eta$-long).
Moreover, the
property that a given term $M$ inhabits a certain type is an
extrinsic property of $M$ and types should play no role when
propagating substitutions. Since the type of a term and its free
variables might not be known in general, it is thus not always
possible to put $x$ in $\beta\eta$-long normal form. We therefore
allow extending substitutions with renamings of variables, such as in
$\sigma' = \sigma, \ctxvarex x$, which maps $x$ to itself.

Applications resemble the way substitutions are built. Using the
typing rules as a guide, notice that being able to apply some atomic
term $R$ to some other term must mean the type of $R$ denotes a
function whose type must be of the shape $\fun{x^n{:}A[\Psi^n]}B$. The
function $R$ then, expects an argument of type $A[\Psi^n]$, which can
therefore only be of the shape $\Psihat^n.M$. It would make little
sense to simply write $M$ here, since $M$ may contain ``free''
variables from $\Psihat^n$. Recall that all occurrences of the
variables $x^n$ in $M$ are associated with a postponed substitution
$\sigma$ which will provide instantiations for the variables in
$\Psi^n$. To further bring home the connection to substitutions,
consider $\beta$-reduction. Eliminating a redex $\app {(\lam{x^n}{N})}
{(\Psihat^n.M)}$ means substituting $\Psihat^n.M$ for $x^n$ in $N$,
i.e. $[\Psihat^n.M/x^n]N$.

\subsection{Context operations}

Before moving to the typing rules, we explain the two
necessary context manipulating operations: merging and chopping. When
checking the domain of a dependent function $\fun{x^n{:}A[\Phi^n]}B$ in
context $\Psi$, it will be necessary to drop some of the assumptions
in $\Psi$ and extend $\Psi$ with $\Phi^n$. To chop off all variables below level
$n$ from the context $\Psi$, we write $\ctxrestrict{\Psi}n$. To merge two
contexts $\Psi$ and $\Phi$ we write $\ctxinsert{\Psi}{\Phi}$.


As mentioned earlier, contexts must be sorted according to the level
of assumptions $x^n{:}A[\Psi^n]$. One should conceptualize this
ordered context as a stack of subcontexts, one for each level of
variables. Let $\Psi(k)$ be the subcontext of $\Psi^n$ with only
assumptions of level $k$. Then, $\Psi^n = \Psi(n-1), \Psi(n-2),
\ldots, \Psi(1)$. We opt here for a flattened presentation of this
stack of contexts in order to simplify merging and chopping of stacks.

However, keeping the context sorted comes at a cost: inserting new
assumptions $x^k{:}A[\Phi^k]$ must respect the invariant that contexts
are always sorted. The flipside is that guaranteeing that merging two
contexts respects well-formedness is much easier and chopping contexts
is more efficient. With merging defined, insertion of a new assumption
into a context is a special case, so we dispense with defining a
separate operation. Merging and chopping contexts are defined
inductively as follows:

\[
\begin{array}{lcll}
\multicolumn{3}{l}{\mbox{Merging contexts:} \ctxinsert{\Psi}{\Phi} = \Gamma }\\
    \ctxinsert{\edot}{\Phi} &= &\Phi \\
    \ctxinsert{\Psi}{\edot} &= &\Psi \\
    \ctxinsert{\Psi, x^n {:} A[\Gamma^n]}{\Phi, y^k {:} B[{\Gamma'}^k]} &= &
    (\ctxinsert{\Psi, x^n {:} A[\Gamma^n]}{\Phi}), y^k {:} B[{\Gamma'}^k]
    &\text{if $k  \leq n$} 
    \\
    \ctxinsert{\Psi, x^n {:} A[\Gamma^n]}{\Phi, y^k {:} B[{\Gamma'}^k]} &=&
    (\ctxinsert{\Psi}{\Phi, y^k {:} B[{\Gamma'}^k]}), x^n {:} A[{\Gamma}^n]
    &\text{otherwise}
\\[1em]
\multicolumn{3}{l}{\mbox{Chopping contexts:} \quad\ctxrestrict{\Psi}{n}~ =~  \Phi
}\\
    \ctxrestrict{\cdot}n &= & \cdot \\ \relax
    \ctxrestrict{(\Psi, x^k {:} A[\Phi^k])}n &= & \ctxrestrict{\Psi}n
    &\text{if $k < n$}
    \\
    \ctxrestrict{(\Psi, x^k {:} A[\Phi^k])}n &= & \Psi, x^k {:} A[\Phi^k]
    &\text{otherwise}
  \end{array}
\]
%

Merging of two independent contexts is akin to the merge step of the
mergesort algorithm and therefore inherits many of its properties. In
particular, the merge of two sorted independent contexts is again a
sorted context. It is also stable, in the sense that the relative
positions of any two assumptions in $\Psi^n$ or in $\Phi^k$ is
preserved in $\ctxinsert{\Psi^n}{\Phi^k}$.

The chopping operation allows us to drop all variable assumptions below a given
index from a context. If $k \leq n$, then $\ctxrestrict{\Psi^k}{n} = \cdot$.
Similar to the chopping and merging operation on contexts, we will need chopping
and merging on the level of simultanous substitutions, written
$\ctxrestrict{\sigma}{\Gamma^n}$ and $\ctxinsert{\sigma}{\rho}$ respectively. These operations will be defined in Section \ref{sec:sub} on page \pageref{page:chopsub}.

\subsection{Typing rules}

We present in this section a bi-directional type system, capable of
checking normal terms (resp. normal types) against a type (resp.
sorts) and synthesizing types (resp. sorts) for atomic entities. The
rules are given in Figure~\ref{fig:typing-rules}. When reading the
rules bottom-up, assumptions are accumulated into the context $\Psi$
at the left of the turnstile as we descend into multi-level objects,
but it is sometimes necessary to restrict it using the previously
defined operations. All typing judgments have access to a well-typed
signature $\Sigma$ where we store constants together with their types
and kinds. However, signatures declare global constants and never
change in the course of a typing derivation. Therefore the
parameterization of the typing rules by the signature $\Sigma$ for
constants is kept implicit.

\begin{center}
\begin{tabular}{ll}
$\Psi \der M \chk A$ & Normal term $M$ checks against type $A$ \\
$\Psi \der R \syn A$ & Neutral term $R$ synthesizes type $A$ \\
$\Psi \der \sigma \chk \Phi^n$ & Substitution $\sigma$ has domain $\Phi^n$ and range $\Psi$. \\[1em]
\end{tabular}
\end{center}

The bi-directional rules can be understood as determining two mutually
defined algorithms for inferring the type of an object and checking an
object against a type. We always assume that $\Psi$ and the subject
($M$, $R$, or $\sigma$) are given, and that the contexts $\Psi$
contains only canonical types and is well-formed. For checking $M \chk
A$ we also assume $A$ is given and canonical, and similarly for
checking $\sigma \chk \Phi$ we assume $\Phi$ is given and is
well-formed. For synthesis $R \syn A$ we assume $R$ is given and we
generate a canonical $A$. Similarly, at the level of types and contexts we have 
\begin{center}
\begin{tabular}{ll}
$\Psi    \der A \chk s$ & Type/Kind $A$ is well-formed \\
$\Psi    \der A \syn K$ & Type $A$ synthesizes kind $K$\\
$\Psi    \der \Phi \jctx$ & Context $\Phi$ is well-formed in the context $\Psi$
\end{tabular}
\end{center}
with corresponding assumptions on the constituents.

As in Pure Type Systems, a type is well formed if its type is a sort.
Whereas signatures might contain term-level and type-level constant declarations,
we only allow declarations of sort $\lftype$ in contexts since abstractions
and dependent function types may only abstract over terms, not types. 

Checking that types are well-kinded is bi-directional. To check that $\fun{x^n{:}A[\Phi^n]}B$ is a well-formed type in the context $\Psi$, we
check first that the context $\Phi^n$ is well-formed in $\Psi$. We note that the
assumptions in $\Phi^n$ should only have access to assumptions
greater than or equal
to $n$ and checking that $\Phi^n$ is well-formed in the context $\Psi$ will
amount to checking that $\Phi^n$ only depends on assumptions $\ctxrestrict{\Psi}{n}$.
Next, we verify that $A$ is well-kinded. Because of the dependency of
types on terms, $\Phi^n$ scopes over the type $A$. Consider for instance
\[
(\hat\Gamma.\mathsf{cons}\;n\;x\;xs) : (\mathsf{vec}\;n)[\Gamma]
\]
where $\Gamma = n{:}\mathsf{nat},x{:}\mathsf{bool},xs{:}\mathsf{vec}\;n$. The type of this
instantiation for a meta-variable depends on the variable $n$ bound in
$\Gamma$. $A$ may refer to the variables in $\Phi^n$, but also to
variables $m$ where $m \geq n$ from $\Psi$. Hence, we drop from $\Psi$
all assumptions below $n$ and merge the resulting context with
$\Phi^n$. Finally, we check that $B$ is well-kinded in the context
$\Psi$ extended with the assumption $x^n{:}A[\Phi^n]$. We rely on the
previously defined merging operation on contexts, to insert the
assumption $x^n{:}A[\Phi^n]$ at the appropriate position in $\Psi$.

\begin{figure}
\centering
  \begin{rules}
    \mbox{\bf Atomic Types/Kinds}\quad\fbox{$\Psi \der P \syn K$}
    \hfill\\[0.5em]
    \infer{\strut{}\Psi \der \const a \syn K}
    {\Sigma(\const a) = K}
    \quad
    \infer{\strut{}\Psi \der \app P{(\hat\Phi^n.N)} \syn [\Phihat^n.N/x^n]_{A[\Phi^n]}K}
    {\Psi \der P \syn \fun{x^n{:}A[\Phi^n]}K & \ctxinsert{\ctxrestrict\Psi{n}}{\Phi^n} \der N \chk A}
    \\[0.5em]
    \mbox{\bf Normal Types/Kinds}\quad\fbox{$\Psi \der A \chk s$}
    \hfill\\[0.1em]
    \infer{\strut{}\Psi \der \lftype \chk \lfkind}{}
    \quad\quad\quad
    \infer{\strut{}\Psi \der P \chk \lftype}
    {\Psi \der P \syn \lftype}
    \\[0.7em]
    \infer{\strut{}\Psi \der \fun{x^n{:}A[\Phi^n]}B \chk s}
    {\ctxinsert{\ctxrestrict{\Psi}{n}}{\Phi^n} \der A \chk \lftype &
      \Psi \der \Phi^n \jctx &
      \ctxinsert\Psi{x^n{:}A[\Phi^n]} \der B \chk s}
    \\[0.5em]
    \mbox{\bf Atomic Terms}\quad\fbox{$\Psi \der M \syn A$}
    \hfill\\[0.5em]
    \infer{\strut{}\Psi \der x^n[\sigma] \syn [\sigma]_{\Phi^n}A}
    {\Psi(x^n) = A[\Phi^n] & \Psi \der \sigma \chk \Phi^n}
    \quad
    \infer{\strut{}\Psi \der \const c \syn A}
    {\Sigma(\const c) = A}
    \quad
    \infer{\strut{}\Psi \der \app R{(\hat\Phi^n.N)} \syn [\Phihat^n.N/x^n]_{A[\Phi^n]}B}
    {\Psi \der R \syn \fun{x^n{:}A[\Phi^n]}B & \ctxinsert{\ctxrestrict\Psi{n}}{\Phi^n} \der N \chk A}
    \\[0.5em]
    \mbox{\bf Normal Terms}\quad\fbox{$\Psi \der M \chk A$}
    \hfill\\[0.5em]
    \infer{\strut{}\Psi \der R \chk Q}
    {\Psi \der R \syn P & P = Q}
    \quad\quad\quad
    \infer{\strut{}\Psi \der \lam{x^n}M \chk \fun{x^n{:}A[\Phi^n]}B}
    {\ctxinsert\Psi{x^n{:}A[\Phi^n]} \der M \chk B}
    \\[0.5em]
    \mbox{\bf Substitutions}\quad\fbox{$\Psi \der \sigma \chk \Phi^n$}
    \hfill\\[0.1em]
    \infer{\strut{}\Psi \der \edot \chk \edot}{}
    \quad\quad
    \infer[\strut{}\mbox{where}\;\sigma' =
    \ctxinsert{\ctxrestrict{\sigma}{k}}{\id(\Gammahat^k)}]
    {\strut{}\Psi \der \sigma, \hat \Gamma^k.M \chk \Phi^n, x^k{:}A[\Gamma^k]}
    {\Psi \der \sigma \chk \Phi^n & 
     \ctxinsert{\ctxrestrict\Psi{k}}{[\sigma]_{\Phi^n}(\Gamma^k)} \der M \chk
     [\sigma']_{(\ctxinsert{\ctxrestrict{\Phi^n}{k}}{\Gamma^k})} A} 

    \\[0.7em]
    \quad\quad\quad\quad\quad\quad\quad
    \infer{\strut{}\Psi \der \sigma, \ctxvarex y^k \chk \Phi^n, x^k{:}A[\Gamma^k]}
    {
      \Psi \der \sigma \chk \Phi^n &
      \Psi(y^k) = [\sigma]_{\Phi^n}(A[\Gamma^k])}
    \\[0.5em]
    \mbox{\bf Context well-formedness}\quad\fbox{$\Psi \der \Phi^n \jctx$}
    \hfill\\[0.5em]
    \infer{\strut{} \Psi \der \edot \jctx}{}
    \quad
    \infer[k<n]{\strut{}\Psi \der \Phi^n, x^k{:}A[\Gamma^k] \jctx}
    {\Psi \der \Phi^n \jctx &
       \ctxinsert{\ctxinsert{\ctxrestrict{\Psi}{n}}
                 {\ctxrestrict{\Phi^n}{k}}}
                 {\Gamma^k}  \der A \chk \lftype &
      \ctxinsert{\ctxrestrict{\Psi}{n}}
      {\ctxrestrict{\Phi^n}{k}}\der \Gamma^k \jctx}

  \end{rules}
  \caption{Typing rules for LF with contextual variables and context variables}
  \label{fig:typing-rules}
\end{figure}

To check that atomic types are well-kinded, we synthesize their kind. For type
constants, we simply look up their type in the signature $\Sigma$. The
interesting case is the application rule. To synthesize the kind for $\app P
(\Phihat^n.N)$, we first synthesize the kind for $P$ as $\fun
{x^n{:}A[\Phi^n]}K$. Subsequently, we check that $N$ has type $A$. We again must
be careful regarding the context. First, some renaming may be necessary to bring
the locally bound variables described in $\Phihat^n$ in sync with the
context. Moreover, we observe that all variables below $n$ which occur in $N$
and $A$ refer to binding sites in $\Phi^n$. All variables equal or greater than $n$
which occur in $N$ and $A$ refer to binding sites in $\ctxrestrict{\Psi}{n}$,
i.e. the context $\Psi$ where we drop all assumptions below $n$. Finally, we
must be careful to substitute $\Phihat^n.N$ for $x^n$ in $K$ in the resulting
kind we return. Because our grammar only recognizes $\beta$-normal objects as
syntactically well-formed, we must rely on hereditary substitution to
hereditarily eliminate any redices as we instantiate variables in the
target of the substitution. We annotate here the substitution
with the type of $A[\Phi^n]$. This is only strictly necessary to ensure that
hereditary substitutions terminate. We postpone the definition and discussion on
hereditary substitutions for now and will revisit it in Section~\ref{sec:sub}.

In the lambda-abstraction rule, we check that $\lambda x^n.M$ has type $\fun
{x^n{:}A[\Phi^n]}B$ by inserting the new assumption $x^n{:}A[\Phi^n]$ at the
appropriate position in $\Psi$ and continuing to check that $M$ has type
$B$. Note that, without loss of generality, we implicitly assume here
and everywhere else that $x^n \notin \Psi$. This can always be achieved by $\alpha$-renaming.
When we reach a normal object of atomic type, we synthesize a type $Q$ and
compare $Q$ to the expected type $P$. Comparing two types reduces to checking
structural equality between $Q$ and $P$ modulo renaming, since all types and terms are always in
canonical form. The only minor complication arises when checking that
two substitutions are equivalent. Because we may simply write $x^n$ for a
variable of type $A[\Psihat^n]$ or its expanded form, comparing two
substitutions must take into account $\eta$-contraction.

To synthesize the type of a constant, we simply look up its type in the
signature $\Sigma$. Term-level application $\app R (\Phihat^n.N)$ follows the
same ideas as type-level applications. The most interesting rule is the one for
variables. To synthesize the type of a variable $x^n$ we retrieve its type
$A[\Phi^n]$ from $\Psi$. Next, we check that the substitution $\sigma$ which is
associated with $x^n$ maps variables from $\Phi^n$ to $\Psi$. Finally, we return
the type of $x^n[\sigma]$ which is $[\sigma]_{\Phi^n}A$.

A substitution $\sigma , \hat\Gamma^k.M$
checks against domain $\Phi^n, x^k{:}A[\Gamma^k]$, if $\sigma$ checks against
$\Phi^n$ and in addition $M$ is well-typed. As in the rules for applications, we
must be a little careful about where variables in $M$ are bound. $M$ contains
locally bound variables from $\hat\Gamma^k$ as well as global variables from
$\ctxrestrict{\Psi}{k}$. We again restrict $\Psi$ to only contain variables
above $k$, since all variables below $k$ are bound in $\Gamma$. Next, we inspect
the type dependencies. We note that $\Gamma^k$  is a well-formed context in
$\Phi^n$, although the typing rules will ensure $\Gamma^k$ only accesses
declarations from $\ctxrestrict{\Phi^n}{k}$. Similarly, when applying
$\sigma$ to $\Gamma^k$, we will ensure that $\sigma$ will be appropriately
restricted (see the definition in the appendix) to only substitute for variables
of level $k$ and higher. Therfore, $[\sigma]_{\Phi^n}(\Gamma^k)$ yields a
well-formed context in $\ctxrestrict{\Psi^n}{k}$.  
On the other hand, $A$ is well-typed in the context
$\ctxinsert{\ctxrestrict{\Phi^n}{k}}{\Gamma^k}$, however $\sigma$ has domain
$\Phi^n$. Simply applying $\sigma$ to $A$ would be incorrect; instead, we
restrict $\sigma$ to contain only the mappings for the variables in
$\ctxrestrict{\Phi^n}{k}$ and map all the variables from $\Gamma^k$ to
themselves. This is written as
$[\ctxinsert{\ctxrestrict{\sigma}{k}}{\id(\Gammahat^k)}]_{(\ctxinsert{\ctxrestrict{\Phi^n}{k}}{\Gamma^k})}$.  

Checking the extension $\sigma, \ctxvarex x^k$ of a substitution by a
variable involves looking up the
declared type of $x^k$ in $\Psi$ and compare it to the expected type. Because
the expected type $A[\Gamma^k]$ was well-typed in $\Phi^n$, we must verify that
$\Psi(x^k) = [\sigma]_{\Phi^n}(A[\Gamma^k])$. Note that
$[\sigma]_{\Phi_n}(A[\Gamma^k]) =
([\ctxrestrict{\sigma}{\Gamma^k}]_{(\ctxinsert{\ctxrestrict{\Phi_n}{k}}{\Gamma^k})}A)[[\sigma]_{\Phi_n} \Gamma^k]$.

Finally, we consider the rules that characterize well-formed contexts.
In the typing rules discussed above, we are often given contexts
$\Phi$ that are not closed but rather whose assumptions might depend
on the ambient context $\Psi$. Since $\Psi$ is already assumed
well-formed, we keep it to the left of the turnstile and write
$\Psi \vdash \Phi^n \jctx$ to mean $\Phi$ is a well-formed context at
level $n$ in context $\Psi$. An alternative would have been to have
judgements of the form $\vdash \Gamma^n \jctx$ only and state that $\Phi^n$
in $\Psi$ is well-formed as $\vdash \ctxinsert{\ctxrestrict
  {\Psi}{n}}{\Phi^n}$. A context $\Phi^n, x^k{:}A[\Gamma^k]$ is well-formed if $\Phi^n$ is
well-formed and $A[\Phi^k]$ is well-typed. Again we must be careful
about the dependency structure. The context $\Gamma^k$ can refer to
variables from $\Phi^n$, but only at levels $k \leq n$. Moreover, any
variable $x^m$ where $n \leq m$ is declared in $\Phi$.

\subsection{Properties}

We begin by proving some properties about contexts and context merging and 
chopping. We first show that we can always increase the upper bound of a context.

\begin{lemma}[Cumulativity]
  If $\Psi \der \Phi^n \jctx$ and $n < k$ then $\Psi \der \Phi^k \jctx$.
\end{lemma}
\LONGVERSION{
\begin{proof}
  By induction on the derivation of $\Psi\der\Phi^n\jctx$.
\end{proof}}

Next, we show that merging produces well-formed contexts, if both contexts are
independent. 

\begin{lemma}[Closure under independent context merging]\mbox{}\\
  If $\cdot \der \Psi^n \jctx$ and $\cdot \der \Phi^k \jctx$ then
  $\cdot \der (\ctxinsert{\Psi^n}{\Phi^k})^{\max(n,k)} \jctx$.
\end{lemma}
\LONGVERSION{
\begin{proof}
  Since they are well-formed, assumptions in $\Psi^n,\Phi^k$ are
  well-ordered with respect to variable levels in the domains. This
  lemma follows from the invariants respected by the merge step in the
  merge-sort algorithm. The merge of two sorted lists is a sorted list
  and merging is stable. Moreover, again since $\Psi^n,\Phi^k$
  well-formed, for all $x^m \in \dom(\Psi^k)$ and $y^{m'} \in
  \dom(\Phi^k)$, $m < n$ and $m' < k$. Therefore
  $\ctxinsert{\Psi^n}{\Phi^k}$ is well-formed at level $\max(n,k)$.
\end{proof}
}
More importantly, if we extend a context $\Psi^n$ with a context $\Phi^k$ where
$\Psi^n \der \Phi^k \jctx$, the resulting context $\ctxinsert{\Psi^n}{\Phi^k}$
is well-formed. This lemma is crucial to ensure that we work with well-formed
contexts during typing.  

\begin{lemma}[Well-formed context extension]\mbox{}\\
  If $\der \Psi^n \jctx$ and $\Psi^n \der \Phi^k \jctx$
  then $\edot \der (\ctxinsert{\Psi^n}{\Phi^k})^{\max(n,k)} \jctx$.
\end{lemma}
\LONGVERSION{
\begin{proof}
By lexicographic induction on the structure of  $\Phi^n$ and $\Psi^k$.
\paragraph{Case 1} $\Phi^k = \cdot$\\
$\der \Psi^n \jctx$ \hfill by assumption \\
$\Psi^n \der \cdot \jctx$ \hfill by assumption \\
$\ctxinsert{\Psi^n}{\cdot} = \Psi^n$ \hfill by definition\\
\textbf{Subcase 1.1:}$k > n$ \\
$\der \Psi^k$ \hfill by cumulativity \\
{\textbf{Subcase 1.2:} $k \leq n$} \\
{$\der \Psi^n \jctx$} \hfill by assumption 

\paragraph{Case 2} $\Phi^k = \Phi'^k, x^l{:}A[\Gamma^l]$\\
{\textbf{Subcase 2.1:} $\Psi^n = \cdot$}\\
$\cdot^n \der \Phi'^k, x^l{:}A[\Gamma^l]\jctx$ \hfill by assumption \\
$\ctxinsert{\cdot^n}{\Phi^k} = \Phi^k$ \hfill by definition \\
{\textbf{Subcase 2.1.1:}$k > n$} \\
{$\der \Phi^k$} \hfill by assumption \\
{\textbf{Subcase 2.1.2:} $k \leq n$} \\
{$\der \Psi^n \jctx$} \hfill by cumulativity 
\\[1em]
{\textbf{Subcase 2.2:} $\Psi^n = \Psi'^n, y^m{:}C[\Gamma'^m]$} \\
$\Psi^n \der \Phi'^k, x^l{:}A[\Gamma^l]\jctx$ \hfill by assumption \\
$\Psi^n \der \Phi'^k \jctx$ \hfill by inversion\\
{\textbf{Subcase 2.2.1:}$m < l $} \\
$\ctxinsert{(\Psi'^n, y^m{:}C[\Gamma'^m])}{ \Phi'^k, x^l{:}A[\Gamma^l]} =
 (\ctxinsert {\Psi'^n}{(\Phi'^k, x^l{:}A[\Gamma^l])}), y^m{:}C[\Gamma'^m]$\hfill by
definition \\
$\der \ctxinsert {\Psi'^n}{(\Phi'^k, x^l{:}A[\Gamma^l)]} \jctx$ \hfill by i.h. \\ 
$\der \Psi'^n, y^m{:}C[\Gamma'^m] \jctx$ \hfill by assumption \\
$\der \Psi'^n \jctx$  \hfill \\
$\Psi'^n \der \Gamma'^m \jctx$ \hfill \\
$\ctxinsert{\ctxrestrict{\Psi'^n }{m}}{\Gamma'^m} \der C \chk \lftype$ \hfill by
inversion \\
$\ctxinsert{\Psi'^n}{\Phi^k} \der \Gamma'^m \jctx$ \hfill by weakening\\
$\ctxinsert{\ctxrestrict{(\ctxinsert{\Psi'^n}{\Phi^k})}{m}}{\Gamma'^m} \der C
\chk \lftype$ \hfill by weakening \\
$\der (\ctxinsert {\Psi'^n}{\Phi^k}), y^m{:}C[\Gamma'^m] \jctx$
\\[0.5em]
{\textbf{Subcase 2.2.2:}$m \geq l $} \\ 
$\ctxinsert{(\Psi'^n, y^m{:}C[\Gamma'^m])}{ \Phi'^k, x^l{:}A[\Gamma^l]} =
 (\ctxinsert {\Psi'^n, , y^m{:}C[\Gamma'^m]}{\Phi'^k}), x^l{:}A[\Gamma^l]$ \hfill by
definition \\
$\der \ctxinsert{\Psi^n}{\Phi'^k} \jctx$ \hfill by i.h. \\
$\der \Phi'^k, x^l{:}A[\Gamma^l] \jctx$ \hfill by assumption \\
$\der \Phi'^k \jctx$  \hfill \\
$\Phi'^k \der \Gamma^l \jctx$ \hfill \\
$\ctxinsert{\ctxrestrict{\Phi'^k }{l}}{\Gamma^l} \der A \chk \lftype$ \hfill by
inversion \\
$\ctxinsert{\Psi^n}{\Phi'^k} \der \Gamma^l \jctx$ \hfill by weakening\\
$\ctxinsert{\ctxrestrict{(\ctxinsert{\Psi^n}{\Phi'^k})}{l}}{\Gamma^l} \der A
\chk \lftype$ \hfill by weakening \\
$\der (\ctxinsert {\Psi'^n}{\Phi'^k}), x^l{:}A[\Gamma^l] \jctx$


\end{proof}
}

\begin{lemma}[Closure under chopping]
  If $\edot \der \Psi^k \jctx$
  then $\edot \der (\ctxrestrict{\Psi}n)^k \jctx$.
\end{lemma}
\LONGVERSION{
\begin{proof}
  By induction on the derivation of $\edot \der \Psi^k \jctx$.
\end{proof}}

\begin{lemma}[Weakening, Identity]\mbox{}
  \begin{enumerate}
  \item If $\Psi \der J$ then $\ctxinsert{\Psi}{x^n{:}A[\Phi^n]} \vdash J$.
  \item Let $\Psi(n) = \wvec{x^n{:}A[\Phi^n]}$ denote the subcontext
    $\Psi(n) \subseteq \Psi$ of assumptions at level $n$. We have that
    $\Psi \der \Psi(n)$.
  \end{enumerate}
\end{lemma}
\LONGVERSION{
\begin{proof}
  By induction on the derivation of $\Psi\der J$ and closure of
  well-formed contexts under merging.
\end{proof}}

\begin{lemma}[Well-formedness of contexts at level $k$]
$\Psi \der \Phi^k \jctx$  iff $\ctxrestrict{\Psi}{k} \der \Phi^k \jctx$.
\end{lemma}
\LONGVERSION{
\begin{proof}
  For a context $\Phi$ at level $n$ in $\Psi$, by induction on the
  derivation of the well-formation judgement $\Psi \der \Phi^n \jctx$,
  remarking that chopping is idempotent:
  $\ctxrestrict{\ctxrestrict{\Phi}{n}}{n} = \ctxrestrict{\Phi}{n}$.
\end{proof}
}


\newcommand{\limplies}{\rightarrow}
\newcommand{\darrow}{\Rightarrow}
\newcommand{\erasex}[1]{{#1}^{-}}
\newcommand{\ldot}{.\,}

\subsection{Hereditary Substitution}\label{sec:sub}

Normal terms are not closed under vanilla substitution, a rather
problematic matter of fact given that our syntax can only express
normal forms. For example, when replacing naively $x$ by $\lambda
y.c\;y$ in the object $x\;z$, we would obtain $(\lambda y.c\;y)\;z$.
It is essential therefore to iron out as we go any redices we might
create as a result of substituting terms for variables. We hence
follow \cite{Watkins02tr} in defining a {\em hereditary} substitution,
which does just that. That hereditary substitutions always terminate
on well-typed normal terms is crucial to ensuring that our typing
rules are decidable. In the above example, hereditary substitutions
continue to substitute $z$ for $y$ in $c\;y$ to obtain $c\;z$ as a
final result. This idea scales to our setting, but we must be careful
to observe the scope of variables.

Hereditary substitution are defined structurally considering the term
to which the substitution operation is applied and the type of the object which
is being substituted. The type is only needed to construct evidence of
termination. We define the hereditary substitution operations for
types, normal object, neutral objects, substitutions, and contexts.

In the formal development it is simpler if we can stick to the structure
of the example above and use only non-dependent types in hereditary
substitutions.  This suffices because we only need to know whether we have
encountered a function which can be reduced further or whether we have reached
an object of base type and reduction will terminate. We therefore first define
type approximations $\alpha$ and an erasure operation $\erasex{()}$ that removes
dependencies.  Before applying any hereditary substitution
$[\Phihat^n.N/x^n]_{A[\Phi^n]}(M)$ we first erase dependencies to obtain
$\alpha[\phi] = \erasex{(A[\Phi])}$ and then carry out the
hereditary substitution proper as $[\Phihat^n.N/x^n]_{\alpha[\phi]}B$.  A
similar convention applies to the other forms of hereditary substitutions.

\[
\begin{array}{rrcl}
\mbox{Type approximation} & \alpha, \beta & ::= & a 
\mid \alpha[\gamma] \darrow \beta \\
\mbox{Context approximation} & \gamma, \psi & ::= & \cdot \mid \gamma, x^n{:}\alpha[\phi]
\mid \gamma, x^n{:}\_
\end{array}
\]
The last form of context approximation, $x^n{:}\_$ is needed when the
approximate type of $x^n$ is not available.\footnote{See the definition
of $[\sigma]_{\psi^n}(\lam{y^n}{M})$ in the electronic appendix.}  It
does not arise directly from erasure.

Types and contexts are related to type and context approximations via
an erasure operation $()^-$ which we overload to work on types
and contexts.
\[
\begin{array}{lcll}
\erasex{(a)} & = & a \\
\erasex{(\app{P}{(\Psihat.N)})} & = & \erasex{P} \\
\erasex{(\Pi{x^n{::}A[\Psi^n]}\ldot B)} & = & \erasex{(A)}[\erasex{(\Psi^n)}] \darrow \erasex{(B)} \\[1ex]
\erasex{(\cdot)} & = & \cdot \\
\erasex{(\Psi^n, x^k{:}A[\Phi^k])} & = & \erasex{(\Psi^n)}, x^k{:}\erasex{(A)}[\erasex{(\Phi^k)}]
\end{array}
\]

Herediatary substitution is given by the following equations. We
overload the substitution operation to work on normal terms, neutral terms,
substitutions, and contexts.

 \[ \begin{array}{ll}
  [\Phihat^n.N/x^n]_{\alpha[\phi]}(N) = N' & \mbox{Hereditary substitution into $N$} \\ \relax
  [\Phihat^n.N/x^n]_{\alpha[\phi]}(R) = R' \quad \mbox{or} \quad M' : \alpha' &
  \mbox{Hereditary substitution into $R$} \\ \relax
  [\Phihat^n.N/x^n]_{\alpha[\phi]}(\sigma) = \sigma' & \mbox{Hereditary
    substitution composition}\\ \relax
  [\Phihat^n.N/x^n]_{\alpha[\phi]}(\Psi) = \Psi' & \mbox{Hereditary substitution into $\Psi$}
 \end{array} \]

Applying a substitution to a neutral term, may yield either a neutral term or a
normal term together with a type approximation. In the latter case, we can
simply drop the type approximation, since it is only necessary for guaranteeing
that any reductions triggered when applying the substitution to a neutral term
will terminate.

\paragraph{Substitution into normal and neutral terms}
We present hereditary substitution for normal terms and neutral terms,
substitutions and context. The definitions for types can be found in the
appendix.

\begin{figure}[htb]
  \centering
\[
\begin{array}{lclp{7.5cm}}
\multicolumn{4}{l}{\mbox{Substitution into normal terms}}\\[0.25em]
\msub{\Phihat^n.N/x^n}_{\alpha[\phi]} (\lambda y^k.M) & = & \lambda y^k. M' & where
$\msub{\Phihat^n.N/x^n}_{\alpha[\phi]}(M) = M'$ if $k < n$
\\[0.25em]
\msub{\Phihat^n.N/x^n}_{\alpha[\phi]} (\lambda y^k.M) & = & \lambda y^k. M' & where
$\msub{\Phihat^n.N/x^n}_{\alpha[\phi]}(M) = M'$ \\
&&& if $y^k \not\in \FV(\Phihat^n.N)$ and $k \geq n$
\\[0.25em]
\msub{\Phihat^n.N/x^n}_{\alpha[\phi]} (R) & = & M
       & \mbox{if $[\Phihat^n.N/x^n]_{\alpha[\phi]}(R) = M' : \alpha'$}  \\ \relax
 \msub{\Phihat^n.N/x^n}_{\alpha[\phi]}(R) & = & R'
       & \mbox{if $[\Phihat^n.N/x^n]_{\alpha[\phi]}(R) = R'$} \\ \relax
 [\Phihat^n.N/x^n]_{\alpha[\phi]}(N) &  & \mbox{fails} & \mbox{otherwise}
\\[0.75em]
\multicolumn{4}{l}{\mbox{Substitution into  neutral terms}}\\[0.5em]\relax
\msub{\Phihat^n.N/x^n}_{\alpha[\phi]} (\const{c}) & = & \const{c} \\
\msub{\Phihat^n.N/x^n}_{\alpha[\phi]} (x^n[\sigma])  & = & N' : \alpha
& where $\msub{\Phihat^n.N/x^n}_{\alpha[\phi]}(\sigma) = \sigma'$ \\
&&& and $[\sigma']_{\phi}(N)= N' $
\\[0.25em]
\msub{\Phihat^n.N/x^n}_{\alpha[\phi]} (y^k[\sigma])  & = &
y^k[\sigma'] & where $\msub{\Phihat^n.N/x^n}_{\alpha[\phi]}(\sigma) = \sigma'$ if $y^k \neq x^n$
\\[0.25em]
\msub{\Phihat^n.N/x^n}_{\alpha[\phi]}(\app R (\Psihat^k.M)) & = & \app {R'}\;(\Psihat^k.M') & where
$\msub{\Phihat^n.N/x^n}_{\alpha[\phi]}(R) = R'$ \\
&&& and $\msub{\Phihat^n.N/x^n}_{\alpha[\phi]}(M) = M' $
if $k \leq n$
\\[0.25em]
\msub{\Phihat^n.N/x^n}_{\alpha[\phi]} (\app R (\Psihat^k.M)) & = & \app {R'}\;(\Psihat^k.M) & where
$\msub{\Phihat^n.N/x^n}_{\alpha[\phi]}(R) = R'$ if $k > n$
\\[0.25em]
\msub{\Phihat^n.N/x^n}_{\alpha[\phi]}(\app R (\Psihat^k.M)) & = & N'' : \beta  & where
$\msub{\Phihat^n.N/x^n}_{\alpha[\phi]}(R) \!=\! {\lam {y^k}{N'}{:} {\gamma[\psi]
    \rightarrow \beta}}$\\
&&& and
     $M' = \msub{\Phihat^n.N/x^n}_{\alpha[\phi]}(M)$ \\
&&& and $N'' = [\Psihat^k.M'/y^k]_{\gamma[\psi]}(N')$ \\
&&& if $\gamma[\psi] \rightarrow \beta \leq \alpha[\phi]$ and $k \leq n$
\\[0.25em]
\msub{\Phihat^n.N/x^n}_{\alpha[\phi]}(\app R (\Psihat^k.M)) & = & N'' : \beta  &
where $\msub{\Phihat^n.N/x^n}_{\alpha[\phi]}(R) \!=\! {\lam {y^k}{N'}{:} {\gamma[\psi] \rightarrow \beta}}$ \\
&&& and $N'' = [\Psihat^k.M/y^k]_{\gamma[\psi]}(N')$ \\
&&& if $\gamma[\psi] \rightarrow \beta \leq \alpha[\phi]$ and $k > n$ \\
\msub{\Phihat^n.N/x^n}_{\alpha[\phi]}(R) & & \mbox{fails} & \mbox{otherwise}
\end{array}
\]
  \caption{Hereditary substitution on normal terms and neutral terms}
  \label{fig:msub}
\end{figure}
%


We define $[\Phihat^n.N/x^n]_{\alpha[\phi]}(M)$ and
$[\Phihat^n.N/x^n]_{\alpha[\phi]}(R)$ by nested induction, first on the structure of the type approximation
$\alpha[\phi]$ and second on the structure of the objects $N$ and $R$.  In
other words, we either go to a smaller type approximation (in which
case the objects can become larger), or the type approximation remains
the same and the objects become smaller.  We write $\alpha \leq \beta$
and $\alpha < \beta$ if $\alpha$ occurs in $\beta$ (as a proper
subexpression in the latter case).  Such occurrences can be inside a
context approximation $\psi$ in the function approximation
$\beta_1[\psi]\darrow \beta_2$, so we also write $\alpha < \psi$ if
$\alpha \leq \beta$ for some $y^k{:}\beta[\gamma]$ in $\psi$, and we write
$\alpha < \beta[\psi]$ if $\alpha \leq \beta$ or $\alpha < \psi$.

When defining
substitutions, we must be careful to take into account where 
multi-level variables are bound. For example, when applying
$\msub{\Phihat^n.N/x^n}$ to a lambda-abstraction $\lambda x^k.M$, we must check
for possible capture, if $k \geq n$. Recall that $\Phihat^n.N$ only binds
variables up to level $n$ locally and $N$ can still refer to variables greater
or equal to $n$ which have from $N$'s perspective a global status. Therefore, if
$k \geq n$, we must ensure that $y^k$ does not occur in the free variables of
$\Phihat^n.N$, written as $\FV(\Phihat^n.N)$. If $k < n$, then $y^k$ can in fact
not appear in $\Phihat^n.N$ because all variables of level $k$ are bound in
$\Phihat^n$. Recall that all variables $x^n[\sigma]$ exist as closures, and hence all
variables in $\Phihat^n$ will be substituted for using $\sigma$.

When considering the substitution operation on neutral terms, two cases are
interesting, applying the substitution to a variable and to an application.
When we apply $\Phihat^n.N/x^n$ to a variable $y^k[\sigma]$, we apply it to
$\sigma$ obtaining $\sigma'$ as a result. If $y^k \neq x^n$, we simply return
$y^k[\sigma']$.  If $y^k = x^n$, i.e. $n = k$ and $y = x$, then we must continue
to apply $\sigma'$ to $N$ obtaining $N'$. We then return $N' : \alpha$, because
$y^k[\sigma]$ may have occurred in a functional position, and we must trigger a
$\beta$-reduction step, if it is applied.

Propagating the hereditary substitution $[\Phihat^n.N/x^n]$ through an application
$\app R {(\Psihat^k.M)}$ is split into multiple cases considering the level and
the  possible elimination of created redices. If $k \leq n$, we need to apply
the substitution not only to $R$ but also to $(\Phihat^k.M)$, because $M$ may
refer to $x^n$; otherwise, we only need to apply the substitution to $R$,
because all occurrences of $x^n$ in $M$ are bound locally by $\Phihat^k$.

If applying the substitution to $R$ produces a normal term $\lambda y^k.N'$,
then we must continue to substitute and replace $y^k$ with $(\Psihat^k.M)$ in
$N'$. The approximate type annotations on the substitution guarantees that the
approximate type of the lambda-abstraction is smaller than the approximate type
of $x^n$, and hence this hereditary substitution will terminate.

\paragraph{Single substitution into simultaneous substitutions}
Applying $\Phihat^n.N/x^n$ to a
substitution $\sigma$ is done recursively and is straightforward, when we keep
in mind the fact that all variables below $k$ are bound within $M$ when we
encounter $\Psihat^k.M$. If $k \geq n$, then applying the substitution
$\Phihat^n.N/x^n$ to $\Psihat^k.M$ will leave it unchanged. Only if $k \leq n$, we
push the substitution $\Phihat^n.N/x^n$ through $M$. When we encounter $\sigma',
\ctxvarex y^k$ where $y^k \neq x^n$, we simply apply the substitution to
$\sigma'$. If we encounter $\sigma', \ctxvarex x^n$, then we must replace $x^n$
by $\Phihat^n.N$.

\[
\begin{array}{lclp{7.5cm}}
\msub{\Phihat^n.N/x^n}_{\alpha[\phi]}(\cdot) & = & \cdot
\\
\msub{\Phihat^n.N/x^n}_{\alpha[\phi]}(\sigma , \Psihat^k.M) & = & \sigma' ,\Psihat^k.M' &
where $\msub{\Phihat^n.N/x^n}_{\alpha[\phi]}(M) = M'$ \\
&&& and $\msub{\Phihat^n.N/x^n}_{\alpha[\phi]}(\sigma) = \sigma'$ if $k \leq n$
\\
\msub{\Phihat^n.N/x^n}_{\alpha[\phi]}(\sigma , \Psihat^k.M) & = & \sigma' ,\Psihat^k.M &
where $\msub{\Phihat^n.N/x^n}_{\alpha[\phi]}(\sigma) = \sigma'$ if $k > n$
\\
\msub{\Phihat^n.N/x^n}_{\alpha[\phi]}(\sigma, \ctxvarex x^n) & = & \sigma' , \Phihat^n.N & where $\msub{\Phihat^n.N/x^n}_{\alpha[\phi]}(\sigma) = \sigma'$
\\
\msub{\Phihat^n.N/x^n}_{\alpha[\phi]}(\sigma, \ctxvarex y^k) & = &
\sigma', \ctxvarex y^k & where
$\msub{\Phihat^n.N/x^n}_{\alpha[\phi]}(\sigma) = \sigma'$ if $y^k \neq x^n$
\end{array}
\]
%

\paragraph{Substitution into contexts}
Applying $\Phihat^n.N/x^n$ to a context
proceeds recursively on the structure of the context until we encounter a
declaration $x^k$ where $k > n$. Because our contexts are ordered, we know that
the remaining context cannot contain an occurrence of $x^n$.

\[
\begin{array}{lclp{7.5cm}}
\msub{\Phihat^n.N/x^n}_{\alpha[\phi]}(\cdot) & = & \cdot \\
\msub{\Phihat^n.N/x^n}_{\alpha[\phi]}(\Psi, y^k{:}A[\Gamma^k]) & = & \Psi',
y^k{:}A'[\Gamma'^k] & where $\msub{\Phihat^n.N/x^n}_{\alpha[\phi]} (\Psi) = \Psi'$
\\
&&&and $\msub{\Phihat^n.N/x^n}_{\alpha[\phi]}(A[\Gamma^k]) = A'[\Gamma'^k]$ if $k \leq n$
\\
\msub{\Phihat^n.N/x^n}_{\alpha[\phi]}(\Psi, y^k{:}A[\Gamma^k]) & = & \Psi,
y^k{:}A[\Gamma^k] & \mbox{if $k > n$}
\end{array}
\]

\paragraph{Simultaneous substitutions} Similar to the single substitution
operation, the simultanous substitution operation is indexed with the domain
of $\sigma$ described by the approximate context $\phi$. To define simultaneous
substitutions it is worth recalling how they arise: A simultaneous substitution
$\sigma$ is associated with a variable $x^n$ which has type $A[\Phi^n]$. As a
consequence, $\sigma$ provides instantiations for
variables declared in $\Phi^n$, i.e. variables up to level $n$. Any variable
at level $n$ or above is bound in the ambient context and $\sigma$ will not
provide any substitutions for such variables. It thus makes sense to
write any upper bound on a simultaneous substitution explicitly as
in $\sigma^n$, just as for contexts, but we generally choose to omit it.
Defining the simultaneous substitution we must be careful to ensure
that the usual substitution properties for simultaneous substitution
holds --- if $\Phi \der \sigma \chk \Psi^n$ and $\ctxrestrict{\Gamma}{n}\ctxmerge\Psi^n
\der J$ then $\ctxrestrict{\Gamma}{n}\ctxmerge\Phi \der
[\sigma]_{\psi}J$. $\ctxrestrict{\Gamma}{n}$ is the ambient context which
remains untouched by the simultanous substitution $\sigma$.

The typing rules act again as a guide in our definitions. For lambda-abstractions for instance, pushing a
simultanous substition $\sigma^n$ through $\lambda y^k.M$ we need to distinguish
cases based on the level: if $k < n$, then we must extend the substitution
$\sigma$ with the identity mapping for $y^k$. From the typing rule for
lambda-abstractions, we see that contexts are ordered and we insert the new
declaration $y^k$ at its correct position using context merging. Similarly, we
need to define substitution merging operation which will extend $\sigma$ and
insert $y^k$ at its correct position. If $k \geq n$, then $\sigma$ does not need
to be extended as we push $\sigma^n$ through the $\lambda$-abstraction because
the substituion $\sigma^n$ has no effect on variables above level $n$.

For applications for instance, applying a simultaneous substitution
$\sigma$ where $\Gamma \der \sigma \chk \Psi^n \jctx$ should take
$\app{R}{(\Phihat^m.N)}$ from the context $\Psi$ to the context
$\Gamma$. This is justified in a straightforward manner by appealing
to the induction hypothesis on the premises of the typing rule.
However, we cannot directly appeal to the induction hypothesis in the
second premise since the assumptions do not match the domain of
$\sigma$. If we want the substitution property to hold, we must define
substitution chopping operation similar the context chopping operation
together with an identity substitution for mapping the variables in
$\Phi^n$ to themselves. We therefore define chopping as
$\ctxrestrict{(\sigma/\psi)}{n}$ inductively on the structure of
$\sigma$ and its domain $\psi$. Both for chopping off parts of a
substitution and merging substitutions, we will resurrect the domain
of the substitution to have access to the level of each variable for
which the substitution provides a mapping. For convenience, we omit
writing out the resurrected contexts during merging and chopping when
they are understood. We also write $F$ for $\ctxvarex{x^n}$ and
$\Psihat^n.N$.

\label{page:chopsub}
\[
\begin{array}{ll@{\;}rcll}
\multicolumn{5}{l}{\mbox{Merging substitutions:} \ctxinsert{\sigma/\psi}{\tau/\phi} = \rho}\\
\multicolumn{3}{l}{\ctxinsert{\edot/\edot}{\tau/\phi}} &= &\tau \\
\multicolumn{3}{l}{\ctxinsert{\sigma/\psi}{\edot/\edot}} &= &\sigma \\
\multicolumn{3}{l}{\ctxinsert{(\sigma/\psi,~F/x^n)}{(\tau/\phi, ~F'/y^k)}} &= &
    (\ctxinsert{(\sigma/\psi, ~F/x^n)}{\tau/\phi}), F'
    &\text{if $k  \leq n$} 
    \\
\multicolumn{3}{l}{\ctxinsert{(\sigma/\psi,~ F/x^n)}{(\tau/\phi, ~F'/y^k)}} &= &
    (\ctxinsert{\sigma/\psi}{(\tau/\phi,~F'/y^k)}), F
    &\text{otherwise}
\\[1em]
\multicolumn{5}{l}{\mbox{Chopping substitutions:} 
\quad\ctxrestrict{(\sigma/\psi)}{n}~ =~  \tau\hspace{6cm}
}\\
\quad\quad\quad\quad\quad & 
   \cdot & \ctxrestrict{}{n} &= & \cdot\\ \relax
 &    (\sigma / \psi, & \ctxrestrict{F /y^k)} {n} &= &
    \ctxrestrict{(\sigma / \psi)}{n}
    &\text{if $k < n$}\\
 &    (\sigma/\psi, &  \ctxrestrict{F/y^k)}{n} &= & \sigma, F
     &\text{if $k \geq n$}
    \\
  \end{array}
\]

The identity substitution, written as $\id(\Phihat)$ is the simple unrolling of
$\Phihat$ into a substitution: $\id(\cdot) = \cdot$ and $\id(\Phihat,x^l) =
\id(\Phihat),\ctxvarex{x^l}$. 
Due to lack of space we omit the definition of simultanous substitution here
(see appendix for the full definition).

\paragraph{Properties of substitutions}
If the original term is not well-typed, a hereditary substitution,
though terminating, cannot always return a meaningful term. In that case, we
simply fail to return a result. Later we show that on well-typed
terms, hereditary substitution always returns well-typed terms.


Applying the substitution to an object will terminate because either we apply
the substitution to a sub-expression or the objects we substitute are
smaller. The following substitution property holds for types, terms,
substitutions and contexts.

\begin{lemma}[Termination] \mbox{}
  \begin{enumerate}
  \item If $\msub{\Psihat^n.N/x^n}_{\alpha[\psi]}(R) = M' : \beta$ then $\beta \leq \alpha[\psi]$.
  \item $\msub{\Psihat^n.N/x^n}_{\alpha[\psi]}(\_)$ terminates, either by returning a result or failing after a finite
number of steps.
  \end{enumerate}
\end{lemma}

\LONGVERSION{
\begin{lemma} \mbox{}
  \label{thm:ssubsts}
  For any term, type or kind $F$,
  \begin{enumerate}
  \item $\msub{\msub\sigma_\phi\rho}_{\phi'}(\msub{\sigma'}_{\phi'} F)
      = \msub\sigma_\phi (\msub\rho_\gamma F)$ where $\sigma' =
      \ctxrestrict{\sigma}{m}\ctxmerge\id(\hat\Gamma^m)$ and $\Phi'
      = (\ctxrestrict{\phi}m \ctxmerge \Gamma^m)$;
    \item $\msub{\sigma}_\phi(\msub{\rho}_\gamma F) =
      \msub{~\ctxrestrict{\sigma}{m} \ctxmerge \msub{\sigma}_\phi\rho}_{\phi'}(F)$ where
    $\Phi'  = (\ctxrestrict{\phi}m \ctxmerge \Gamma^m)$;
  \item $\msub\sigma_\phi([\hat\Gamma.N/x^n] F) = \msub{\sigma \ctxmerge \msub\sigma_\phi(\hat\Gamma.N)/x^n}F$;
  \item $[\sigma]_\phi([\Gammahat^n.N/x^n]F) = [\Gammahat^n.N/x^n]([\sigma]_\phi F)$ if level$(\phi) < n$.
  \end{enumerate}
\end{lemma}
\begin{proof}
  By induction on $F$.
\end{proof}
}

\begin{lemma}[Identity extension]\mbox{}\\
  \label{thm:identity-extension}
  \begin{enumerate}
  \item 
    If $\Psi \der \sigma \chk \Phi$
    then
    $\Psi \ctxmerge x^k{:}[\sigma]_{\Phi}(A[\Gamma^k]) \der
    \rho \chk \Phi \ctxmerge x^k{:}A[\Gamma^k]$
    where $\rho = \sigma/\phi \ctxmerge \ctxvarex x^k/x^k$.
  \item
    If $\Psi \der \sigma \chk \Phi$
    then
    $\Psi \ctxmerge \msub\sigma_\phi \Gamma \der
    \rho \chk \Phi \ctxmerge \Gamma$
    where $\rho = \sigma/\phi \ctxmerge \id(\gamma)$.
  \end{enumerate}
\end{lemma}
\LONGVERSION{
\begin{proof}
  \begin{enumerate}
  \item 
  By induction on the derivation of $\Psi\der\sigma\chk\Phi$. Let
  $\Psi' = \Psi \ctxmerge x^k{:}[\sigma]_\Phi(A[\Gamma^k])$. $\Psi'$ is
  a well-formed context by Lemma~\ref{thm:wf-merge}.

  \paragraph{Case: $\Psi\der\edot\chk\edot$:}
  $\sigma = \edot$ so $\Psi'=\Psi\ctxmerge x^k{:}A[\Gamma^k]$ by
  definition. Noting that $\rho = \ctxvarex x^k$ by definition of
  merging of substitutions and $\edot$ is the neutral element of
  merging, we then obtain the result by the following derivation:
  \[
  \infer{\Psi' \der \edot, \ctxvarex x^k \chk \edot, x^k{:}A[\Gamma^k]}
  {\infer{\Psi'\der\edot\chk\edot}{} & \Psi'(x^k) = [\edot]_\edot A[\Gamma^k]}
  \]

  \paragraph{Case: $\Psi\der\sigma,\hat\Gamma^n.M \chk \Phi, y^n{:}B[\Gamma^n]$:} \mbox{} \\
  \[
  \infer{
    \Psi \der \sigma, \hat\Gamma^n.M \chk \Phi, y^n{:}B[\Gamma'^n]}
  {\Psi \der \sigma \chk \Phi
  & \ctxrestrict{\Psi}{n} \ctxmerge [\sigma]_{\phi}(\Gamma'^n) \der M \chk [\sigma']_{\phi'}(B)}
  \]
\mbox{where $\sigma' = \ctxrestrict{\sigma}{n} \ctxmerge \id(\Gamma^n)$ and 
            $\Phi' = \ctxrestrict{\Phi}{n} \ctxmerge \Gamma$}
        \\
let $\sigma_w = \ctxinsert{\sigma}{\ctxvarex{x^k}}$ and $\Phi_w =
\ctxinsert{\Phi}{x^k{:}A[\Gamma^k]}$\\
and similarly, $\sigma'_w = \ctxinsert{\sigma'}{\ctxvarex{x^k}}$ and 
$\Phi'_w = \ctxinsert{\Phi'}{x^k{:}A[\Gamma^k]}$. \\
  $\ctxrestrict{\Psi}{n} \ctxmerge [\sigma_w]_{\phi_w}\Gamma'^n
  \der M \chk [\sigma'_w]_{\phi'_w} B$
  \hfill by substitution weakening \\
  $\Psi\ctxmerge x^k{:}[\sigma]_{\Phi}(A[\Gamma^k]) \der
  \sigma_w \chk \Phi_w$  \hfill by I.H. \\
  $\Psi\ctxmerge x^k{:}[\sigma, \Gamma'^k.M]_{\Phi_w}(A[\Gamma^k]) \der
  \sigma_w \chk \Phi_w$
  \hfill by substitution weakening \\
  \paragraph{}{\bf Sub-case 1:} $n \leq k$\\
  $\ctxrestrict{\Psi}{n} \ctxmerge [\sigma]_{\phi}(\Gamma'^n) \der M \chk
  [\sigma']_{\phi'}(B)$ \hfill derived previously \\
  $\ctxrestrict{(\ctxinsert{\Psi }{x^k{:}A[\Gamma^k]})}{n}  
   \ctxmerge [\sigma]_{\phi}(\Gamma'^n) \der M \chk
  [\sigma']_{\phi'}(B)$ \hfill be weakening \\
  $\ctxrestrict{(\ctxinsert{\Psi }{x^k{:}A[\Gamma^k]})}{n}  
   \ctxmerge [\sigma_w]_{\phi_w}(\Gamma'^n) \der M \chk
  [\sigma'_w]_{\phi'_w}(B)$ \hfill by substitution weakening\\
  $\ctxinsert{\Psi}{x^k{:}A[\Gamma^k]} \der \sigma_w, \Gammahat'^n.M \chk
   \Phi^w, y^n{:}B[\Gamma']$ \hfill by typing rule\\
  $\ctxinsert{\Psi}{x^k{:}A[\Gamma^k]} \der 
\ctxinsert{(\sigma, \Gammahat'.M)}{\ctxvarex{x^k}}
\chk \ctxinsert{(\Phi, y^n{:}B[\Gamma'])}{x^k{:}A[\Gamma]}$ \\
\mbox{$\quad$}\hfill  by definition since $n \leq k $.

\paragraph{}{\bf Sub-case 2:} $n > k$,\\
We note weakening the premise $\ctxrestrict{\Psi}{n} \ctxmerge
[\sigma]_{\phi}(\Gamma'^n) \der M \chk [\sigma']_{\phi'}(B)$ has no effect. We
use directly the typing rule on the I.H. and the premise to derive
\\
  $\ctxinsert{\Psi}{x^k{:}A[\Gamma^k]} \der \sigma_w, \Gammahat'^n.M \chk
   \Phi^w, y^n{:}B[\Gamma']$ \hfill \\
  $\ctxinsert{\Psi}{x^k{:}A[\Gamma^k]} \der 
\ctxinsert{(\sigma, \Gammahat'.M)}{\ctxvarex{x^k}}
\chk \ctxinsert{(\Phi, y^n{:}B[\Gamma'])}{x^k{:}A[\Gamma]}$ \\
\mbox{$\quad$}\hfill  by definition since $n \leq k $.

  \paragraph{Case $\Delta\der\sigma',\ctxvarex y^n \chk \Psi',
    y^n{:}B[\Gamma'^n]$:} \mbox{} \\
  \[
  \infer{
    \Psi \der \sigma, y^n \chk \Phi, y^n{:}B[\Gamma'^n]}
  {\Psi \der \sigma \chk \Phi
  & \Psi(y^n) = [\sigma]_{\phi}(B[\Gamma'])}
  \]
let $\sigma_w = \ctxinsert{\sigma}{\ctxvarex{x^k}}$ and $\Phi_w =
\ctxinsert{\Phi}{x^k{:}A[\Gamma^k]}$\\
$\ctxinsert{\Psi}{x^k{:}[\sigma]A[\Gamma^k]} \der \sigma_w \chk \Phi_w$ \hfill by I.H.\\
$\ctxinsert{\Psi}{x^k{:}[\sigma_w]A[\Gamma^k]} \der \sigma_w \chk \Phi_w$ \hfill
by substitution weakening\\
$(\ctxinsert{\Psi}{x^k{:}[\sigma_w]A[\Gamma^k]})(y^n) =
[\sigma]_{\phi}(B[\Gamma'])$ \hfill   since $\Psi(y^n) = [\sigma]_{\phi}(B[\Gamma'])$\\
$[\sigma]_{\phi}(B[\Gamma']) = [\sigma_w]_{\phi_w}(B[\Gamma'])$ \hfill by
properties of substitution \\ 
$\ctxinsert{\Psi}{x^k{:}[\sigma_w]A[\Gamma^k]} \der \sigma_w, y^n\chk \Phi_w,
y^n{:}B[\Gamma']$ \hfill by typing rule \\
$\ctxinsert{\Psi}{x^k{:}[\sigma_w]A[\Gamma^k]} \der \ctxinsert{(\sigma,
  y^n)}{x^k} \chk \ctxinsert{(\Phi, y^n{:}B[\Gamma'])}{x^k{:}A[\Gamma]}$ \hfill by definition

  \item By induction on $\Gamma$, using part 1.
  \end{enumerate}
\end{proof}
}

\begin{lemma}[Substitution property] \mbox{}
  \begin{enumerate}
  \item
    If $\ctxrestrict\Delta{n} \ctxmerge \Psi^n \der \sigma \chk \Phi^n$
    and $\ctxrestrict\Delta{n} \ctxmerge \Phi^n \der J$
    then $\ctxrestrict\Delta{n} \ctxmerge \Psi^n \der
    \msub{\sigma}_\phi (J)$.
  \item
    If $\ctxrestrict\Psi{n} \ctxmerge \Phi^n \entails N \chk A$
    and $\Psi, x^n{:}A[\Phi^n] \entails J$
    then $\Psi \entails \msub{\Phihat^n.N/x^n}_{\alpha[\phi]} (J)$.
  \end{enumerate}
  Note that in the case of type synthesis of terms, the conclusion of
  the statements depend on the result of the substitution, e.g.:
  \begin{enumerate}[(a)]
  \item $\ctxrestrict\Delta{n} \ctxmerge \Psi^n \der
     R' \syn \msub\sigma_\phi C$ if $\msub\sigma_\phi R = R'$.
  \item $\ctxrestrict\Delta{n} \ctxmerge \Psi^n \der
 M \chk \msub\sigma_\phi C$ if $\msub\sigma_\phi
    R = M : \alpha$ where $\alpha = \erasex{(\msub\sigma_\phi C)}$.
  \end{enumerate}
\end{lemma}
\LONGVERSION{
  \begin{proof}
    We prove the two parts simultaneously, by lexicographic induction
    on the length of the derivation of the second judgement, the level, the type
    approximation $\alpha$, the judgement approximation $\phi$. We
    detail the proof for checking judgements $J = M \chk C$ and
    synthesizing judgements $J = R \syn C $ only, and even then only
    for the non-trivial cases; the proofs for types and substitutions
    are similar. 
    \begin{enumerate}
    \item \strut
      \paragraph{}{\bf Case: $\vcenter{\infer{\strut{}\ctxrestrict\Delta{n}
          \ctxmerge \Phi^n \der x^m[\rho] \syn \msub\rho_\gamma A}
        {(\ctxrestrict\Delta{n} \ctxmerge \Phi^n)(x^m) = A[\Gamma^m] &
          \ctxrestrict\Delta{n} \ctxmerge \Phi^n \der \rho \chk
          \Gamma^m}}$} \mbox{} \\
      $\ctxrestrict\Delta{n} \ctxmerge \Psi^n 
          \der \sigma \chk \Phi^n$ \hfill by assumption \\
      $\ctxrestrict\Delta{n} \ctxmerge \Psi^n 
           \der \msub\sigma_\phi\rho \chk \msub\sigma_\phi\Gamma^m$    \hfill by I.H. \\
      \paragraph{}{\bf Sub-case 1:~}{$x^m \in \ctxrestrict\Delta{n}$ and $m \geq n$}
      \mbox{} \\
      $(\ctxrestrict\Delta{n})(x^m) = A[\Gamma^m]$ \hfill by assumption \\
      $(\ctxinsert{\ctxrestrict\Delta{n}}{\Psi^n})(x^m) = A[\Gamma^m]$ \hfill by definition since $n \leq m$ \\
      $\msub{\sigma}_\phi(\Gamma^m) = \Gamma$ \hfill since $m \geq n = $level($\phi$) \\
      $\ctxrestrict\Delta{n} \ctxmerge \Psi^n
           \der \msub\sigma_\phi\rho \chk \Gamma^m$    \hfill by rewriting in the I.H. \\
      $\ctxrestrict\Delta{n} \ctxmerge \Psi^n \der
      x^m[\msub\sigma_\phi\rho] \syn \msub{\msub\sigma_\phi\rho}_\gamma A$
      \hfill by typing rule\\
      $\ctxrestrict\Delta{n} \ctxmerge \Psi^n \der
      \msub{\sigma}_\phi(x^m[\rho]) \syn \msub\sigma_\phi(\msub{\rho}_\gamma A)$  \hfill by definition

      \paragraph{}{\bf Sub-case 2:~}{$x^m \in \Phi^n$ and $m < n$}
      \mbox{} \\
      $(\Phi^n)(x^m) = A[\Gamma^m]$ \hfill by assumption \\
      $\Phi^n = \Phi_1, x^m{:}A[\Gamma^m],~\Phi_2$ \hfill by previous line \\
      $\ctxrestrict\Delta{n} \ctxmerge \Psi^n 
          \der \sigma \chk \Phi^n$ \hfill recall by assumption \\
      $\ctxrestrict\Delta{n} \ctxmerge \Psi^n 
          \der \sigma_1 \chk \Phi_1$ \hfill for some $\sigma_1/\phi_1 \subseteq \sigma/\phi$ by typing rule for $\sigma$ \\ 
      Either $\msub\sigma_\phi (x^m[\rho])= y^m[\msub\sigma_\phi\rho]$
      or
      $\msub\sigma_\phi (x^m[\rho]) = \msub{\msub\sigma_\phi\rho}_\phi M
      : \alpha'$.

      \paragraph{}{\bf Sub-case 2.1:~}{$\msub\sigma_\phi (x^m[\rho])=
        y^m[\msub\sigma_\phi\rho]$} \mbox{} \\
      $(\Psi^n)(y^m) = \msub{\sigma_1}_{\phi_1}(A[\Gamma^m])$ \hfill by typing rules \\
      $(\Psi^n)(y^m) = \msub{\sigma}_\phi(A[\Gamma^m])$      \hfill by substitution weakening \\
      $(\Psi^n)(y^m) = ([\sigma']_{\phi'} A)[[\sigma]_{\phi^n}\Gamma^m]$  \hfill by definition \\
      \mbox{$\quad$} \hfill where $\sigma' =  \ctxrestrict{\sigma}{m}\ctxmerge\id(\hat\Gamma^m)$ and
      $\Phi' = (\ctxrestrict{\Phi}m \ctxmerge \Gamma^m)$ \\
      $\ctxrestrict\Delta{n} \ctxmerge \Psi^n 
           \der \msub\sigma_\phi\rho \chk \msub\sigma_\phi\Gamma^m$    \hfill recall the I.H. \\
      $\ctxinsert{\ctxrestrict{\Delta}{n}}{\Psi^n} \der y^m[\msub{\sigma}_\phi \rho] 
         \syn [\msub{\sigma}_\phi \rho]_\gamma([\sigma']_{\phi'}A)$ \hfill by typing rule \\
      $\ctxinsert{\ctxrestrict{\Delta}{n}}{\Psi^n} \der   (y^m[\msub{\sigma}_\phi\rho] )
         \syn \msub{\sigma}_\phi ([\rho]_\gamma A)$ \hfill by lemma \ref{thm:ssubsts}\\    

      \subparagraph{}{\bf Sub-case 2.2:~}{$\msub\sigma_\phi (x^m[\rho]) = \msub{\msub\sigma_\phi\rho}_\phi M
      : \alpha'$}
    \mbox{} \\
      $\ctxrestrict{(\ctxrestrict\Delta{n} \ctxmerge \Psi^n)}{m} \ctxmerge \msub{\sigma_1}_{\phi_1}\Gamma^m 
           \der M \chk   \msub{\sigma_1'}_{\phi_1'} A$
           \hfill for some $\sigma_1/\phi_1 \subseteq \sigma/\phi$ by typing
                  rule for $\sigma$ \\
      \mbox{$\quad$} \hfill where $\sigma_1' =
      \ctxrestrict{\sigma_1}{m} \ctxmerge\id(\hat\Gamma^m)$ and $\Phi_1'
      = (\ctxrestrict{\Phi_1}m \ctxmerge \Gamma^m)$ \\
      $\ctxrestrict{(\ctxrestrict\Delta{n} \ctxmerge \Psi^n)}{m}
      \ctxmerge \msub{\sigma}_{\phi}\Gamma^m \der M \chk
      \msub{\sigma'}_{\phi'} A$
      \hfill by substitution weakening \\
      \mbox{$\quad$} \hfill where $\sigma' =
      \ctxrestrict{\sigma}{m} \ctxmerge\id(\hat\Gamma^m)$ and $\Phi'
      = (\ctxrestrict{\Phi}m \ctxmerge \Gamma^m)$ \\
      $\ctxrestrict\Delta{n} \ctxmerge \Psi^n 
           \der \msub\sigma_\phi\rho \chk \msub\sigma_\phi\Gamma^m$    \hfill recall previously derived by I.H. \\
      $\ctxrestrict{(\ctxrestrict\Delta{n} \ctxmerge \Psi^n)}{m}
      \der
      \msub{\msub\sigma_\phi\rho}_{\gamma} M \chk
      \msub{\msub\sigma_\phi\rho}_{\gamma} (\msub{\sigma'}_{\phi'} A)$
      \hfill by I.H. since level$(\Gamma) = m < n$\\
      $\ctxrestrict{(\ctxrestrict\Delta{n} \ctxmerge \Psi^n)}{m}
      \der
      \msub{\msub\sigma_\phi\rho}_{\gamma} M 
         \chk \msub{\sigma}_\phi ([\rho]_\gamma A)$ \hfill by lemma \ref{thm:ssubsts}\\    

      \paragraph{}{\bf Case: $\vcenter{\infer{\strut{}\ctxrestrict\Delta{n} \ctxmerge \Phi^n \der \app R{(\hat\Gamma^m.N)} \syn [\hat\Gamma^m.N/x^m]_{A[\Gamma^m]}B}{
          \ctxrestrict\Delta{n} \ctxmerge \Phi^n \der R \syn
          \fun{x^m{:}A[\Gamma^m]}B & \ctxrestrict\Delta{n}
            \ctxmerge \ctxrestrict{\Phi^n}m \ctxmerge \Gamma^m \der N \chk A}}$}

      \paragraph{}{\bf Sub-case 1:~}{$\msub\sigma_\phi R = R'$ is atomic} \mbox{}\\
      $\ctxrestrict\Delta{n} \ctxmerge \Psi^n \der R' \syn \msub\sigma_\phi (\fun{x^m{:}A[\Gamma^m]}B)$
      \hfill by I.H. 

      \paragraph{}{\bf Sub-case 1.1:~}{$m < n$} \mbox{} \\
      $\ctxrestrict\Delta{n} \ctxmerge \Psi^n \der R' \syn
      \fun{x^m{:}{\msub\sigma_\phi(A[\Gamma^m])}}{\msub{\sigma'}_{\phi'} B}$
      \hfill by definition  \\
      \mbox{$\quad$} \hfill where $\sigma' = \ctxinsert{\sigma}{\ctxvarex{x^m}}$ and $\phi' = \ctxinsert{\phi}{x^m{:}\alpha[\gamma]}$\\
      $\ctxinsert{\ctxrestrict{\Delta}{n}}{\Psi^n} \der \ctxrestrict{\sigma}{m} \chk \ctxrestrict{\Phi^n}{m}$ \hfill by typing rule \\       
        $\ctxrestrict\Delta{n}
        \ctxmerge \ctxrestrict{\Psi^n}m \ctxmerge \msub\sigma_\phi(\Gamma^m)\der \tau \chk \Phi^n_0$
      \hfill by Lemma~\ref{thm:identity-extension} \\
      \mbox{$\quad$} \hfill where $\tau = \ctxrestrict{\sigma}{m}
      \ctxmerge \id(\Gamma^m)$ and $\Phi^n_0 = \ctxrestrict{\Phi^n}{m} \ctxmerge \Gamma^m$
      \\
      \mbox{$\quad$} \hfill and noting that $\msub{\sigma}_\phi(\Gamma^m) = \msub{\ctxrestrict{\sigma}{m}}_{\ctxrestrict{\phi}{m}}(\Gamma^m)$. \\
      $\ctxinsert{\ctxrestrict{\Delta}{n}}
                {\ctxinsert{\ctxrestrict{\Psi^n}{m}}{\msub{\sigma}_\phi(\Gamma^m)}}
        \der [\tau]_{\phi_0}N \chk [\tau]_{\phi_0}A$ \hfill by I.H.\\
      $\ctxinsert{\ctxrestrict{\Delta}{n}}{\Psi^n} \der
          R'\;(\Gammahat^m.[\tau]_{\phi_0}N) \syn \msub{\Gammahat^m.[\tau]_{\phi_0}N/x^m}_{\alpha[\gamma]}(\msub{\sigma'}_{\phi,x^m{:}\alpha[\gamma]} B)$ \hfill by typing rule\\
      $\ctxinsert{\ctxrestrict{\Delta}{n}}{\Psi^n} \der
          R'\;(\Gammahat^m.[\tau]_{\phi_0}N) \syn \msub{\sigma}_{\phi}(\msub{\Gammahat^m.N/x^m}_{\alpha[\gamma]}B$ \hfill by Lemma~\ref{thm:ssubsts}\\
      $\ctxinsert{\ctxrestrict{\Delta}{n}}{\Psi^n} \der
          [\sigma]_\phi(R\;(\Gammahat^m.N) \syn \msub{\sigma}_{\phi}(\msub{\Gammahat^m.N/x^m}_{\alpha[\gamma]}B$ \hfill by definition of substitution

      \paragraph{}{\bf Sub-case 1.2:~}{$m \ge n$} \mbox{} \\
      $\ctxrestrict\Delta{n} \ctxmerge \Psi^n \der R' \syn
      \fun{x^m{:}{(A[\Gamma^m])}}{\msub\sigma_\phi B}$
      \hfill by definition \\
      $\ctxinsert{\ctxrestrict{\Delta}{n}}{\Psi^n} \der
          R'\;(\Gammahat^m.N) \syn \msub{\Gammahat^m.N/x^m}_{\alpha[\gamma]}(\msub{\sigma}_\phi B$ \hfill by typing rule \\
      $\ctxinsert{\ctxrestrict{\Delta}{n}}{\Psi^n} \der
          [\sigma]_\phi(R\;(\Gammahat^m.N))\syn (\msub{\sigma}_\phi
          (\msub{\Gammahat^m.N/x^m}_{\alpha[\gamma]} B)$  \hfill by
          definition and Lemma~\ref{thm:ssubsts}

      \paragraph{}{\bf Sub-case 2:~}{$\msub\sigma_\phi R = \lam{y^k}M : \alpha[\gamma] \Rightarrow \beta$ is
        normal} \mbox{} \\
      $\ctxrestrict\Delta{n} \ctxmerge \Psi^n \der
       \lam{y^k}M \chk \msub\sigma_\phi(\fun{x^m{:}A[\Gamma^m]}B)$
      \hfill by I.H.

      \paragraph{}{\bf Sub-case 2.1:~}{$m < n$} \mbox{} \\
      $\ctxrestrict\Delta{n} \ctxmerge \Psi^n \der
       \lam{y^k}M \chk
      \fun{x^m{:}{\msub\sigma_\phi(A[\Gamma^m])}}{\msub{\sigma'}_{\phi'} B}$
      \hfill by definition  \\
      \mbox{$\quad$} \hfill where $\sigma' = \ctxinsert{\sigma}{\ctxvarex{x^m}}$ and $\phi' = \ctxinsert{\phi}{x^m{:}\alpha[\gamma]}$\\
      $\ctxinsert{\ctxrestrict{\Delta}{n}}{\Psi^n} \der \ctxrestrict{\sigma}{m} \chk \ctxrestrict{\Phi^n}{m}$ \hfill by typing rule \\     
      $\ctxrestrict{\Delta}{n} \ctxmerge \Psi^n \ctxmerge y^m{:}A[\Gamma^m] \der M \chk \msub{\sigma'}_{\phi'}B$ and $k = m$ \hfill by inversion on typing \\
        $\ctxrestrict\Delta{n}
        \ctxmerge \ctxrestrict{\Psi^n}m \ctxmerge \msub\sigma_\phi(\Gamma^m)\der \tau \chk \Phi'^n$
      \hfill by Lemma~\ref{thm:identity-extension} \\
      \mbox{$\quad$} \hfill where $\tau = \ctxrestrict{\sigma}{m}
      \ctxmerge \id(\Gamma^m)$ and $\Phi_0^n = \ctxrestrict{\Phi^n}{m} \ctxmerge \Gamma^m$
      \\
      \mbox{$\quad$} \hfill and noting that $\msub{\sigma}_\phi(\Gamma^m) = \msub{\ctxrestrict{\sigma}{m}}_{\ctxrestrict{\phi}{m}}(\Gamma^m)$. \\
      $\ctxinsert{\ctxrestrict{\Delta}{n}}
                {\ctxinsert{\ctxrestrict{\Psi^n}{m}}{\msub{\sigma}_\phi(\Gamma^m)}}
        \der [\tau]_{\phi'}N \chk [\tau]_{\phi_0}A$ \hfill by I.H.\\
        $\ctxrestrict{\Delta}{n} \ctxmerge \Psi^n \der [\Gammahat^m.[\tau]_{\phi'}N / x^m]_{\alpha[\gamma]}M 
           \chk [\Gammahat^m.[\tau]_{\phi'}N / x^m]_{\alpha[\gamma]}([\sigma']_{\phi'} B)$ \hfill by I.H. (since $\alpha[\gamma] < \phi$) \\
           \mbox{$\quad$} \hfill and noting that $\Psi^n$ does not depend on $y^m$ \\
        $\ctxrestrict{\Delta}{n} \ctxmerge \Psi^n \der [\Gammahat^m.[\tau]_{\phi'}N / x^m]_{\alpha[\gamma]}M 
           \chk [\sigma]_{\phi}([\Gammahat^m.N / x^m]_{\alpha[\gamma]} B)$\hfill by Lemma~\ref{thm:ssubsts}\\
        $\ctxrestrict{\Delta}{n} \ctxmerge \Psi^n \der 
           [\sigma]_\phi(R\;(\Gammahat^m.N))            \chk
           [\sigma]_{\phi}([\Gammahat^m.N / x^m]_{\alpha[\gamma]} B)$ \hfill by definition

      \paragraph{}{\bf Sub-case 2.2:~}{$m \ge n$} \mbox{} \\
      $\ctxrestrict\Delta{n} \ctxmerge \Psi^n \der
       \lam{y^k}M \chk
      \fun{x^m{:}{(A[\Gamma^m])}}{\msub\sigma_\phi B}$ \hfill by definition of substitution \\
      $\ctxrestrict\Delta{n} \ctxmerge \Psi^n \ctxmerge x^m{:}A[\Gamma^m] \der
      M \chk \msub\sigma_\phi B$ and $k = m$
      \hfill by inversion \\
      $\ctxrestrict\Delta{n} \ctxmerge \Psi^n \der
      [\hat\Gamma. N / x^m]_{\alpha[\gamma]} M \chk [\hat\Gamma. N / x^m]_{\alpha[\gamma]} (\msub\sigma_\phi B)$
      \hfill by I.H. (since $\alpha[\gamma] < \phi$) \\
      $\ctxrestrict\Delta{n} \ctxmerge \Psi^n \der
      [\hat\Gamma. N / x^m]_{\alpha[\gamma]} M \chk  \msub\sigma_\phi ([\hat\Gamma. N / x^m]_{\alpha[\gamma]} B)$
      \hfill by Lemma~\ref{thm:ssubsts}\\
      $\ctxrestrict{\Delta}{n} \ctxmerge \Psi^n \der 
           [\sigma]_\phi(R\;(\Gammahat^m.N))            \chk
           [\sigma]_{\phi}([\Gammahat^m.N / x^m]_{\alpha[\gamma]} B)$ \hfill by definition

      \paragraph{}{\bf Case $\vcenter{\infer{\strut{}\ctxrestrict\Delta{n} \ctxmerge \Phi^n \der \lam{x^m}M \chk \fun{x^m{:}A[\Gamma^m]}B}
        {\ctxinsert{\ctxrestrict\Delta{n} \ctxmerge
            \Phi^n}{x^m{:}A[\Gamma^m]} \der M \chk B}}$:} \mbox{}\\

      $\ctxrestrict{\Delta}n \ctxmerge \Psi^n
      \ctxmerge x^m{:}[\sigma]_{\phi}(A[\Gamma^m])
      \der \sigma' \chk \Phi'^n$
      \hfill by Lemma~\ref{thm:identity-extension} \\
      \mbox{$\quad$} where $\sigma' = \sigma \ctxmerge \ctxvarex x^m$
      and $\Phi'^n = \Phi^n \ctxmerge x^m{:}A[\Gamma^m]$ \\
      $\ctxinsert{\ctxrestrict\Delta{n} \ctxmerge
            \Psi^n}{x^m{:}\msub\sigma_\phi(A[\Gamma^m])} \der \msub{\sigma'}_{\phi'} M \chk \msub{\sigma'}_{\phi'} B$
      \hfill by I.H. \\
      $\ctxrestrict\Delta{n} \ctxmerge \Psi^n \vdash
      \lam{x^m}{\msub\sigma_\phi M} \chk \fun{x^m{:}\msub\sigma_\phi
        A[\Gamma^m]}{\msub\sigma_\phi B}$
      \hfill by typing rule \\
      $\ctxrestrict\Delta{n} \ctxmerge \Psi^n \vdash
      \msub\sigma_\phi \lam{x^m}{M} \chk \msub\sigma_\phi \fun{x^m{:}A[\Gamma^m]}{B}$
      \hfill by definition

    \item Similarly as for simultaneous substitutions.
    \end{enumerate}
  \end{proof}
}

The typing judgments are syntax-directed and therefore clearly decidable.
Hereditary substitution always terminates, giving us a decision procedure
for dependent typing.
\begin{theorem}[Decidability of Type Checking] \mbox{}\newline
 All judgments in the dependent contextual modal type theory are decidable.
\end{theorem}
\LONGVERSION{
\begin{proof}
 The typing judgments are syntax-directed and therefore clearly decidable.
Hereditary substitution always terminates, giving us a decision procedure
for dependent typing.
\end{proof}
}


\LONGVERSION{
\input{properties}
}
\section{Related Work}

\paragraph{Multi-level logics of contexts}
Contextual reasoning has been extensively studied for various applications in
AI. For example, Giunchiglia {\textit{et~al}} have explored contextual reasoning
\cite{Giunchiglia93} and have investigated a multi-language hierarchical logic
where we have an infinite level of multiple distinct languages.

In \cite{GiunchigliaSerafini:AI94} they introduce a class of multi-language
systems which use a hierarchy of first-order languages, each language containing
names of the language below. Any two adjacent languages in the
hierarchy are linked only by two bridge rules. The hierarchy is understood as an
alternative to extending modal logics with new modalities. Their goal is to
provide a foundation to the implementation of ``intelligent'' reasoning
systems. As the authors observe, we may use a different system to reason about a
object logic (which may rely on induction) vs reasoning within a given object
logic.
Indexes encode information of the ``locality of the reasoning'', where the
reasoning take place. This is similar in spirit to our use of level
annotations on variables.



\paragraph{Multi-level meta-variables}
We motivated the multi-level system in the introduction with the need to model
the dependency of holes on bound variables as well as meta-variables. This
naturally leads to a multi-level system. This idea has played an important role
in Sato et al \cite{Sato:CSL03} where the authors develop a multi-level calculus
for meta-variables. As in our work, variables carry an index to indicate whether
they are a bound variable, meta-variable, or a meta$^2$-variable, etc.
The main difference compared to our work is that the authors define a
``textual'' substitution which allows capture. This is unlike our
capture-avoiding substitution operation. There are two main obstacles with
textual substitutions. First,  we will lose confluence.
The second problem is that some reductions may get stuck. To address
these problems the authors suggest to define reductions in such a way
that it takes into account the different levels and keep track of
arities of functions. This leads to a carefully engineered system
which is confluent and strongly normalizing, although not very
intuitive. We believe our framework is simpler.

Gabbay and Lengrand \cite{Gabbay:LFMTP07,Gabbay:InfComp09} propose a multi-level calculus for
meta-variables called the Lambda Context Calculus where variables are modeled
via nominals. They also define two different kinds of substitutions: one, a
capture-allowing substitution, i.e. a meta-variable of level $n$ is allowed to
capture names below $n$ and two, a captrue-avoiding substitution for variables
greater than $n$. This  inherently leads to difficulaties regarding
confluence. Our work has one uniform capture-avoiding substitution operation
leading to a more elegant calculus.

Finally, we mention the work by Geuvers and Jogjov (see for example
\cite{Geuvers:CSL02}). Open terms are
represented via a kind of meta-level Skolem function. However, in
general reduction and instantiation of meta-variables (or holes) do
not commute. This problem also arises in Bognar and de Vrijer
\cite{Bognar:JAR01}.  Our work resolves many of these aforementioned
problems, since reduction and instantiation naturally commute and require no
special treatment.

\paragraph{Functional multi-level staged computation}
The division of programs into two stages has been studied intensively in partial evaluation and staged functional computation. Davies and Pfenning
\cite{Davies:ACM01} proposed the use of the modal necessity operator to provide
a type-theoretic foundation for staged computation and more specifically,
run-time code generation. Abstractly, values of type $\Box A$ stand for an
(unevaluated) source expression. To avoid generating closures,
contextual types may be used to generate ``open'' (unevaluated) source
expressions \cite{Nanevski:ICML05}.  An open source expression would then have
type $A[\Psi]$ to describe code of type $A$ in a context $\Psi$. 

However, this two-level framework does not allow us to specify multi-level
transition points (e.g. “dynamic until stage $n$”). For example, Gl{\"u}ck and
J{\o}rgensen \cite{Gluck:Ershov96,Gluck:PLILP95} propose a multi-level
specialization to allow an accurate and fast multi-level binding-time analysis.
This means that a given program can be optimized with respect to some inputs at
an earlier stage, and others at later stages. Gl{\"u}ck and J{\o}rgensen
generalized two-level program generators into multi-level ones, called
multi-level generating extensions. By this generalization, a generated code
fragment can be used for different levels. Subsequently, Yuse and Igarashi
\cite{Yuse:PPDP06} proposed a logical foundation for multi-level generating
extensions based on linear time temporal logics. We believe that our
framework of multi-level contextual types may be used as an alternative to
generate open code which can be used at different levels and manage its
dependency on previous levels cleanly. 


\section{Conclusion}
We generalized and extended the original contextual type theory of Nanevski
et al \cite{Nanevski:ICML05}, where we distinguish between meta-variables and
bound variables, to multiple levels. This streamlines the original presentation  
with fewer typing rules, syntax and operations but with many of the same
properties. Substitutions are defined, checked applied and manipulated in
exactly the same way as meta-substitutions, for instance. We believe our
framework provides already a suitable foundation for formalizing contexts in
theorem proving and functional programming. Unlike other attempts
to provide a multi-level calculus, we believe our work which is based on
contextual modal types avoids and simplifies many of the issues which arise such
as capture-avoiding substitution and the related issues of confluence.


While we have used a named calculus for expository
purposes, the system also generalizes nicely a calculus based on de
Bruijn indices and suspensions more typical of an implementation, such
as that of Abel and Pientka \cite{Abel:LFMTP10}. Variables are then
pairs of indexes into a substitution inside a stack of substitutions.

We have shown that it is possible to {\em express} meta$^k$-terms in
this generalized contextual type theory but have left largely
untouched the question of how to attach a computational behaviour to
such objects and {\em compute} with them. A first step towards
representing computations is to move to a non-canonical calculus that
permits arbitrary (typed) terms, to which we could attach arbitrary
rewrite rules as in deduction modulo \cite{dowek:modulo}. We could
then add meaningful recursors to some of the layers to write proofs by
induction, or add more computational effects for a layer acting as a
tactic layer for a programming and reasoning system such as Beluga. We
would obtain a uniform framework for all of representations of syntax,
proofs over these representations and tactics over these proofs. And
indeed, such a system would be useful for implementing Beluga within
Beluga itself.


 \bibliographystyle{eptcs}
\bibliography{references}

\LONGVERSION{
\input{appendix}
}
\end{document}